\documentclass[a4paper,UKenglish,cleveref, autoref, thm-restate]{lipics-v2021}

\usepackage[dvipsnames,usenames]{xcolor}

\RequirePackage{hyperref}
\hypersetup{colorlinks=true, urlcolor=Blue, citecolor=Green, linkcolor=BrickRed, breaklinks, unicode}

\usepackage{tikz}
\usepackage{xspace,enumerate}
\usepackage{todonotes}
\usepackage{MnSymbol}
\usepackage[noadjust]{cite}
\usepackage{thm-restate}

\usepackage[ruled]{algorithm}
\usepackage[noend]{algpseudocode}
\algnewcommand{\LineComment}[1]{\State \(\triangleright\) #1}

\def\dd{\mathinner{.\,.}}
\newcommand{\cO}{\mathcal{O}}
\newcommand{\cOtilde}{\tilde{\cO}}
\newcommand{\LCS}{\texttt{LCS}}

\newcommand{\anc}{\textsf{anc}}
\newcommand{\lev}{\textsf{lev}}

\DeclareMathOperator*{\argmax}{argmax}
\DeclareMathOperator*{\argmin}{argmin}

\newcommand{\TG}{\mathcal{T}}
\newcommand{\AG}{\mathcal{A}}
\newcommand{\Vor}{\textsf{Vor}}
\newcommand{\dist}{\textsf{dist}}
\newcommand{\VD}{\textsf{VD}}
\newcommand{\out}[1]{#1^{out}}

\newcommand{\last}[1]{#1_\searrow}
\newcommand{\med}[1]{\textsf{last}(#1)}
\newcommand{\weight}{{\rm \omega}}
\renewcommand{\top}[1]{\lefthalfcap(#1)}
\renewcommand{\bot}[1]{\righthalfcup(#1)}
\newcommand{\tp}{\lefthalfcap}
\newcommand{\bt}{\righthalfcup}
\newcommand{\DDG}{\textsf{DDG}}
\newcommand{\sep}{\textsf{sep}}

\title{An Almost Optimal Edit Distance Oracle}

\author{Panagiotis Charalampopoulos}{The Interdisciplinary Center Herzliya, Israel}{panagiotis.charalampopoulos@post.idc.ac.il
}{https://orcid.org/0000-0002-6024-1557}{}

\author{Pawe\l{} Gawrychowski}{University of Wroc\l{}aw, Poland}{gawry@cs.uni.wroc.pl}{https://orcid.org/0000-0002-6993-5440}{}

\author{Shay Mozes}{The Interdisciplinary Center Herzliya, Israel}{smozes@idc.ac.il}{https://orcid.org/0000-0001-9262-1821}{}

\author{Oren Weimann}{University of Haifa, Israel}{oren@cs.haifa.ac.il}{https://orcid.org/0000-0002-4510-7552}{}

\authorrunning{P.~Charalampopoulos, P.~Gawrychowski, S.~Mozes, and O.~Weimann}

\Copyright{Panagiotis Charalampopoulos, Pawe\l{} Gawrychowski, Shay Mozes, and Oren Weimann}



\category{} 

\relatedversion{} 

\nolinenumbers 

\hideLIPIcs  

\EventEditors{Nikhil Bansal and James Worrell}
\EventNoEds{3}
\EventLongTitle{48th International Colloquium on Automata, Languages, and Programming (ICALP 2021)}
\EventShortTitle{ICALP 2021}
\EventAcronym{ICALP}
\EventYear{2021}
\EventDate{July 12--16, 2021}
\EventLocation{University of Glasgow, UK (virtual conference)}
\EventLogo{}
\SeriesVolume{}
\ArticleNo{}

\begin{document}

\maketitle

\begin{abstract}
We consider the problem of preprocessing two strings $S$ and $T$, of lengths $m$ and $n$, respectively, in order to be able to efficiently answer the following queries:
Given positions $i,j$ in $S$ and positions $a,b$ in $T$, return the optimal alignment of $S[i \dd j]$ and $T[a \dd b]$. Let $N=mn$. We present an oracle with preprocessing time $N^{1+o(1)}$ and space $N^{1+o(1)}$ that answers queries in $\log^{2+o(1)}N$ time. In other words, we show that we can query the alignment of every two substrings in almost the same time it takes to compute just the alignment of $S$ and $T$. 
Our oracle uses ideas from our distance oracle for planar graphs [STOC 2019] and exploits the special structure of the alignment graph. Conditioned on popular hardness conjectures, this result is optimal up to subpolynomial factors. 
Our results apply to both edit distance and longest common subsequence (LCS).

The best previously known oracle with construction time and size $\cO(N)$ has slow $\Omega(\sqrt{N})$ query time [Sakai, TCS 2019], and the one with size $N^{1+o(1)}$ and query time  $\log^{2+o(1)}N$ (using a planar graph distance oracle) has slow $\Omega(N^{3/2})$ construction time [Long \& Pettie, SODA 2021]. We improve both approaches by roughly a $\sqrt N$ factor.

\end{abstract}

\section{Introduction}

String alignment is arguably the most popular problem in combinatorial pattern matching.
Given two strings $S$ and $T$ of length $m$ and $n$, the problem asks to compute the similarity between the strings according to some similarity measure. The two most popular similarity measures are edit distance and longest common subsequence (LCS). In both cases, the classical solution is essentially the same: Compute the shortest path from vertex $(0,0)$ to vertex $(m,n)$ in the so called {\em alignment graph} of the two strings. 
As taught in almost every elementary course on algorithms, computing this shortest path (and hence the optimal alignment of the two strings) can easily be done in $\cO(N)$ time where $N=mn$, via dynamic programming. Interestingly, this time complexity cannot be significantly improved assuming popular conjectures such as the strong exponential time hypothesis (SETH)~\cite{DBLP:conf/focs/AbboudBW15,DBLP:journals/siamcomp/BackursI18,DBLP:conf/focs/BringmannK15}.
In fact, by now we seem to have a rather good understanding of the complexity of this problem for different similarity measures
and taking other parameters than the length of the both strings into the account, see~\cite{DBLP:conf/focs/BringmannK15}.   

\subparagraph*{Substring queries. } A natural direction after having determined the complexity of a particular problem on strings
is to consider the more general version in which we need to answer queries on substrings of the input string. This has been
done for alignment~\cite{DBLP:journals/corr/abs-0707-3619,DBLP:journals/tcs/Sakai19}, pattern matching~\cite{DBLP:journals/tcs/KellerKFL14,DBLP:conf/soda/KociumakaRRW15}, approximate pattern matching~\cite{unified}, dictionary matching~\cite{DBLP:conf/isaac/Charalampopoulos19,CD_CPM}, compression~\cite{DBLP:journals/tcs/KellerKFL14}, periodicity~\cite{DBLP:conf/spire/KociumakaRRW12,DBLP:conf/soda/KociumakaRRW15}, counting palindromes~\cite{DBLP:conf/spire/RubinchikS17}, longest common substring~\cite{amir19}, computing minimal and maximal suffixes~\cite{DBLP:journals/tcs/BabenkoGKKS16,DBLP:conf/cpm/Kociumaka16}, and computing the lexicographically $k$-th suffix~\cite{DBLP:conf/soda/BabenkoGKS15}.

\subparagraph*{Alignment oracles.}
Consider the shortest path from vertex $(i,j)$ to vertex $(a,b)$ in the alignment graph. It corresponds to the optimal alignment of two substrings: the substring of $S$ between indices $i$ and $j$ and the substring of $T$ between indices $a$ and $b$. 
An {\em alignment oracle} is a data structure that, after preprocessing, can report the optimal alignment score of any two substrings of $S$ and $T$. That is, given positions $i,j$ in $S$ and positions $a,b$ in $T$, the oracle returns the optimal alignment score of $S[i \dd j]$ and $T[a \dd b]$ (or equivalently, the $(i,j)$-to-$(a,b)$ distance in the alignment graph). 

Tiskin~\cite{DBLP:journals/corr/abs-0707-3619,DBLP:journals/jda/Tiskin08} considered a restricted variant of the problem, in which the queries are either the entire string $S$ vs. a  substring of $T$ or a prefix of $S$ vs. a suffix of $T$. 
For such queries, Tiskin gave an $\cOtilde(n+m)$-size oracle, that can be constructed in $\cOtilde(N)$ time, and answers queries in $\cO(\log N/ \log \log N)$ time~\cite{DBLP:journals/corr/abs-0707-3619}. 
For the general problem, Sakai~\cite{DBLP:journals/tcs/Sakai19} (building on Tiskin's work~\cite{DBLP:journals/corr/abs-0707-3619}) showed how to construct in $\cO(N)$ time an alignment oracle with $\cO(n+m)$ query time.
In this work we show that, perhaps surprisingly, obtaining such an oracle can be done essentially for free! That is, at almost the same time it takes to compute just the alignment of $S$ and $T$. More formally, our main result is: 

\begin{theorem}\label{thm:main}
For two strings of lengths $m$ and $n$, with $N=mn$, we can construct in $N^{1+o(1)}$ time an alignment oracle achieving either of the following tradeoffs:
\begin{itemize}
\item $N^{1+o(1)}$ space and $\log^{2+o(1)} N$ query time,
\item $N\log^{2+o(1)} N$ space and $N^{o(1)}$ query time.
\end{itemize}
\end{theorem}

\subparagraph*{Planar distance oracles and Voronoi diagrams.}
The starting point of our work is the recent developments in {\em distance oracles} for planar graphs. A distance oracle is a compact representation of a graph that allows to efficiently query the distance between any pair of vertices. 
 Indeed, since the alignment graph is a planar graph, the state-of-the-art distance oracle for planar graphs of Long and Pettie~\cite{DBLP:journals/corr/abs-2007-08585} (which builds upon~\cite{DBLP:conf/soda/GawrychowskiMWW18,DBLP:conf/stoc/Charalampopoulos19,DBLP:conf/focs/Cohen-AddadDW17}) is an alignment oracle with space $N^{1+o(1)}$ and query time $\cO(\log^{2+o(1)} N)$. However, the construction time of this oracle is $\Omega(N^{3/2})$. Our main contribution  is an improved $N^{1+o(1)}$ construction time when the underlying graph is not just a planar graph but an alignment graph. 

Our oracle has the same recursive structure as the planar graph oracles in~\cite{DBLP:conf/soda/GawrychowskiMWW18,DBLP:conf/stoc/Charalampopoulos19,DBLP:journals/corr/abs-2007-08585} (in fact, the alignment graph, being a grid, greatly simplifies several technical, but standard, difficulties of the recursive structure). These oracles (inspired by Cabello's use of Voronoi diagrams for the diameter problem in planar graphs~\cite{DBLP:journals/talg/Cabello19}) use the recursive structure in order to apply (at different levels of granularity) an efficient mechanism for {\em point location on Voronoi diagrams}.  
At a high level, a {\em Voronoi diagram} with respect to a subset $S$ of vertices (called sites) is a partition of the vertices into $|S|$ parts (called Voronoi cells), where the cell of site $s \in S$ contains all vertices that are closer to $s$ than to any other site in $S$. 
A {\em point location} query, given a vertex $v$, returns the site $s$ such that $v$ belongs to the Voronoi cell of~$s$. Our main technical contribution is a polynomially faster construction of the point location mechanism when the underlying graph is an alignment graph. We show that,  in this case, the special structure of the Voronoi cells facilitates point location via a non-trivial divide and conquer. 
Unlike the planar oracles, which use planar duality to represent Voronoi diagrams, the representation and point location mechanisms we develop in this paper are novel and achieve the same query time, while being arguably simpler than those of Long and Pettie.\footnote{We believe that our efficient construction can also be made to work, for alignment graphs, with the dual representation of Voronoi diagrams used in~\cite{DBLP:conf/soda/GawrychowskiMWW18,DBLP:conf/stoc/Charalampopoulos19,DBLP:journals/corr/abs-2007-08585}, but we think the new representation makes the presentation more approachable as it exploits the structure of the alignment graph more directly.}   

It is common that techniques are originally developed for pattern matching problems (and in particular alignment problems) and later extended to planar graphs.
A concrete example is the use of Monge matrices and unit-Monge matrices. However, it is much less common that  techniques are first developed for planar graphs (in our case, the use of Voronoi diagrams) and only then translated to pattern matching problems. 

\subparagraph*{Conditional lower bounds.}
Any lower bound on the time required to compute an optimal alignment of two strings directly implies an analogous lower bound for the sum of the preprocessing time and the query time of an alignment oracle. In particular, the existence of an oracle for which this sum is $\cO(N^{1-\epsilon})$, for a constant $\epsilon>0$, would refute SETH~\cite{DBLP:journals/siamcomp/BackursI18,DBLP:conf/focs/BringmannK15}.

In the {\em Set Disjointness} problem, we are given a collection of $m$ sets $A_1, A_2, \ldots,A_m$ of total size $M$ for preprocessing. We then need to report, given any query pair $A_i,A_j$, whether $A_i \cap A_j = \emptyset$.
The Set Disjointness conjecture~\cite{DBLP:conf/wads/GoldsteinKLP17,DBLP:journals/siamcomp/PatrascuR14,DBLP:journals/corr/abs-1006-1117} states that any data structure with constant query time must use $M^{2-o(1)}$ space.
Goldstein et al.~\cite{DBLP:conf/wads/GoldsteinKLP17} 
stated the following stronger conjecture.

\begin{conjecture}[Strong Set Disjointness Conjecture~\cite{DBLP:conf/wads/GoldsteinKLP17}]
Any data structure for the Set Disjointness problem that answers queries in time $t$ must use space $M^2/(t^2 \cdot \log^{\cO(1)}M)$.
\label{conj:SSDC}
\end{conjecture}

The following theorem implies that, conditioned on the above conjecture, our alignment oracle is optimal up to subpolynomial factors; its proof is identical to that of~\cite[Theorem 1]{amir19} as explained in~\cite{panosthesis}.

\begin{theorem}[\cite{amir19,panosthesis}]
An alignment oracle for two strings of length at most $n$ with query time $t$ must use $n^2/(t^2 \cdot \log^{\cO(1)}n)$ space, assuming the Strong Set Disjointness Conjecture.\end{theorem}

Even though the main point of interest is in oracles that achieve fast (i.e. constant, polylogarithmic, or subpolynomial) query-time, the above lower bound suggests to study other tradeoffs of space vs.~query-time. 
In~\cref{sec:subquadratic} we show oracles with space sublinear in $N$. More formally, we prove the following theorem. 

\begin{restatable}{theorem}{frbased}\label{thm:frbased}
Given two strings of lengths $m$ and $n$ with $N=mn$, integer alignment weights upper-bounded by $w$, and a parameter $r \in [\sqrt{N},N]$ we can construct in $\cOtilde(N)$ time an $\cOtilde(Nw/ \sqrt{r}+m+n)$-space alignment oracle that answers queries in time $\cOtilde(\sqrt N+r)$.
\end{restatable}

For example, if the alignment weights are constant integers, by setting $r=\sqrt{N}$ we obtain an $\cO(N^{3/4}+m+n)$-space oracle that answers queries in time $\cOtilde(\sqrt N)$.

\subparagraph*{Other related works.}
When the edit distance is known to be bounded by some threshold $k$, an efficient edit distance oracle can be obtained via the Landau-Vishkin algorithm~\cite{LandauV89}.
Namely, after $\cO(n+m)$ time preprocessing of the input strings (oblivious to the threshold $k$), 
given a substring of $S$, a substring of $T$, and a threshold $k$, in $\cO(k^2)$ time one can decide whether the edit distance of these substrings is at most $k$, and if so, return it.

Further, after an $\cO(n+m)$-time preprocessing, given a substring $X$ of $S$ and a substring $Y$ of $T$, the starting positions of substrings of $Y$ that are at edit distance at most $k$ from $X$ can be returned in $\cO(k^4\cdot |Y|/|X|)$ time~\cite{unified}.

In a recent work~\cite{DBLP:conf/cpm/Charalampopoulos20a} on dynamically maintaining an alignment, it was shown that, in the case where the alignment weights are small integers, two strings of total length at most $n$ can be maintained under edit operations in $\cOtilde(n)$ time per operation and per query (here a query asks for the alignment of the two current strings).

\section{Preliminaries}
The alignment oracle presented in this paper applies to both edit distance and longest common subsequence (LCS). To simplify the presentation we focus on LCS but the extension to edit distance is immediate. Namely, our main result (\cref{thm:main}) applies to arbitrary alignment weights and~\cref{thm:frbased}  applies to constant alignment weights.  

The LCS of two strings $S$ and $T$ is a longest string that is a subsequence of both $S$ and $T$. We denote the length of an LCS of $S$ and $T$ by $\LCS(S,T)$.

\begin{example}
An LCS of $S=\texttt{{\color{red}a}c{\color{red}bcd}daa{\color{red}e}a}$ and $T=\texttt{{\color{red}a}bb{\color{red}bc}c{\color{red}de}c}$ is $\texttt{\color{red}abcde}$; $\LCS(S,T)=5$.
\end{example}

For strings $S$ and $T$, of lengths $m$ and $n$ respectively (we will assume that $n \geq m$), the alignment graph $G$ of $S$ and $T$ is a directed acyclic graph of size $N=\cO(mn)$.  
For every $0 \leq x \leq m$ and $0\leq y \leq n$, the alignment graph $G$ has a vertex $(x,y)$ and the following unit-length edges (defined only if both endpoints exist):
\begin{itemize}
\item $((x,y),(x+1,y))$ and $((x,y),(x,y+1))$,
\item $((x,y),(x+1,y+1))$, present if and only if $S[x]=T[y]$.
\end{itemize}

Intuitively, $G$ is an $(m+1)\times (n+1)$ grid graph augmented with diagonal edges corresponding to matching letters of $S$ and $T$. See~\cref{fig:align}. 
We think of the vertex $(0,0)$ as the top-left vertex of the grid and the vertex $(m,n)$ as the bottom-right vertex of the grid.
We shall refer to the rows and columns of $G$ in a natural way. 
It is easy to see that $\LCS(S,T)$ equals $n+m$ minus the length of the shortest path from $(0,0)$ to $(m,n)$ in $G$.

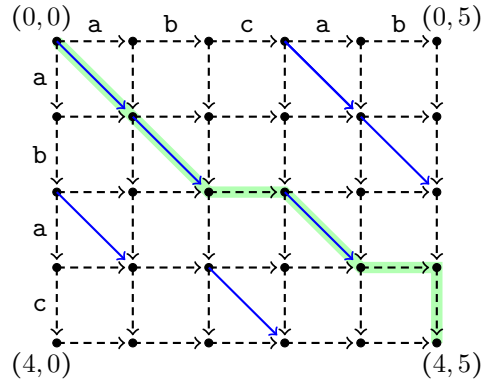
\begin{figure}[ht]
\captionsetup{singlelinecheck=off}
\centering
\begin{tikzpicture}

\draw[green!30!white,line width=1.5mm] (0,4) -- (1,3) -- (2,2) -- (3,2) -- (4,1) -- (5,1) -- (5,0);
  
\foreach \x in {0, 1,2,3,4,5}{
   \foreach \y in {0, 1,2,3,4}{
   \filldraw(\x,\y) circle (1.5pt);
   }
}

\foreach \x/\z in {0/1, 1/2, 2/3, 3/4, 4/5}{
   \foreach \y in {0,1,2,3,4}{
	\draw[->,densely dashed,thick] (\x,\y) -- (\z-0.1,\y);
   }
}

\foreach \x/\z in {0/1, 1/2, 2/3, 3/4}{
   \foreach \y in {0,1,2,3,4,5}{
	\draw[->,densely dashed,thick] (\y,\z) -- (\y,\x+0.1);
   }
}

\foreach \x/\y in {0/4, 3/4, 0/2, 3/2, 4/3, 3/4, 1/3, 2/1}{
	\draw[->,thick, blue] (\x,\y) -- (\x+0.9,\y-0.9);
   }
   
\draw (-0.2,4) node[above] {$(0,0)$};
\draw (5.2,4) node[above] {$(0,5)$};

\draw (5.2,0) node[below] {$(4,5)$};
\draw (-0.2,0) node[below] {$(4,0)$};

\foreach \x/\c in {0.5/a, 1.5/b, 2.5/c, 3.5/a, 4.5/b}{
   \draw (\x,4) node[above] {\large $\texttt{\c}$};
}

\foreach \x/\c in {0.5/c,1.5/a, 2.5/b, 3.5/a}{
   \draw (0,\x) node[left] {\large $\texttt{\c}$};
}
\end{tikzpicture}
\caption[The alignment graph.]{The alignment graph for $S=\texttt{abac}$ and $T=\texttt{abcab}$. 
We represent the horizontal and vertical edges by dashed black arrows, and the diagonal edges by blue arrows.
A lowest scoring $
(0,0)$-to-$(4,5)$ path is highlighted in green, it has weight $6$ and corresponds to the LCS $\texttt{aba}$ of length $3=9-6=|T|+|S|-6$.}\label{fig:align}
\end{figure}

\subparagraph*{Multiple-source shortest paths.} Given a planar graph with $N$ vertices and a distinguished face $h$, the multiple-source shortest paths (MSSP) data structure~\cite{MSSP,DBLP:journals/siamcomp/CabelloCE13} represents all shortest path trees rooted at the vertices of $h$ using a persistent dynamic tree.
It can be constructed in $\cO(N\log N)$ time, requires $\cO(N\log N)$ space, and can report the distance between any vertex $u$ of $h$ and any other vertex $v$ in the graph in $\cO(\log N)$ time.
The MSSP data structure can be augmented at no asymptotic overhead (cf.~\cite[Section 5]{DBLP:journals/talg/KaplanMNS17}), to allow for the following.
First, to report a shortest $u$-to-$v$ path $\rho$ in time $\cO(|\rho| \log\log \Delta)$, where $\Delta$ is the maximum degree of a vertex in $G$.
Second, to support the following queries in $\cO(\log N)$ time~\cite{DBLP:conf/soda/GawrychowskiMWW18}: Given two vertices $u,v \in G$ and a vertex $x$ of $h$ report whether $u$ is an ancestor of $v$ in the shortest path tree rooted at $x$, and whether $u$ occurs before $v$ in a preorder traversal of this tree.
(We consider shortest path trees as ordered trees with the order inherited from the planar embedding.)

\subparagraph*{Recursive decomposition.} We assume without loss of generality that the length of each of the two strings is a power of $2$, and hence the alignment graph is a $(2^a+1)\times (2^b+1)$ grid.
We consider a recursive decomposition $\AG$ of $G$ such that in each level all pieces are of the same rectangular shape.
At each level, each piece will be of size $(2^c +1)\times (2^d +1)$ for non-negative integers $c$ and $d$. 
Consider a piece $P$ of size $(2^c +1)\times (2^d +1)$, with $c+d\neq 0$. 
Assuming without loss of generality that $c\geq d$, in the next level we will partition $P$ to two pieces, each of size $(2^{c-1} +1)\times (2^d +1)$, that share the middle row of $P$. See~\cref{fig:grid}.
We view $\AG$ as a binary tree and identify a piece $P$ with the node corresponding to it in $\AG$.

\begin{figure}[h]
\begin{center}
\includegraphics[page=1,scale=0.35]{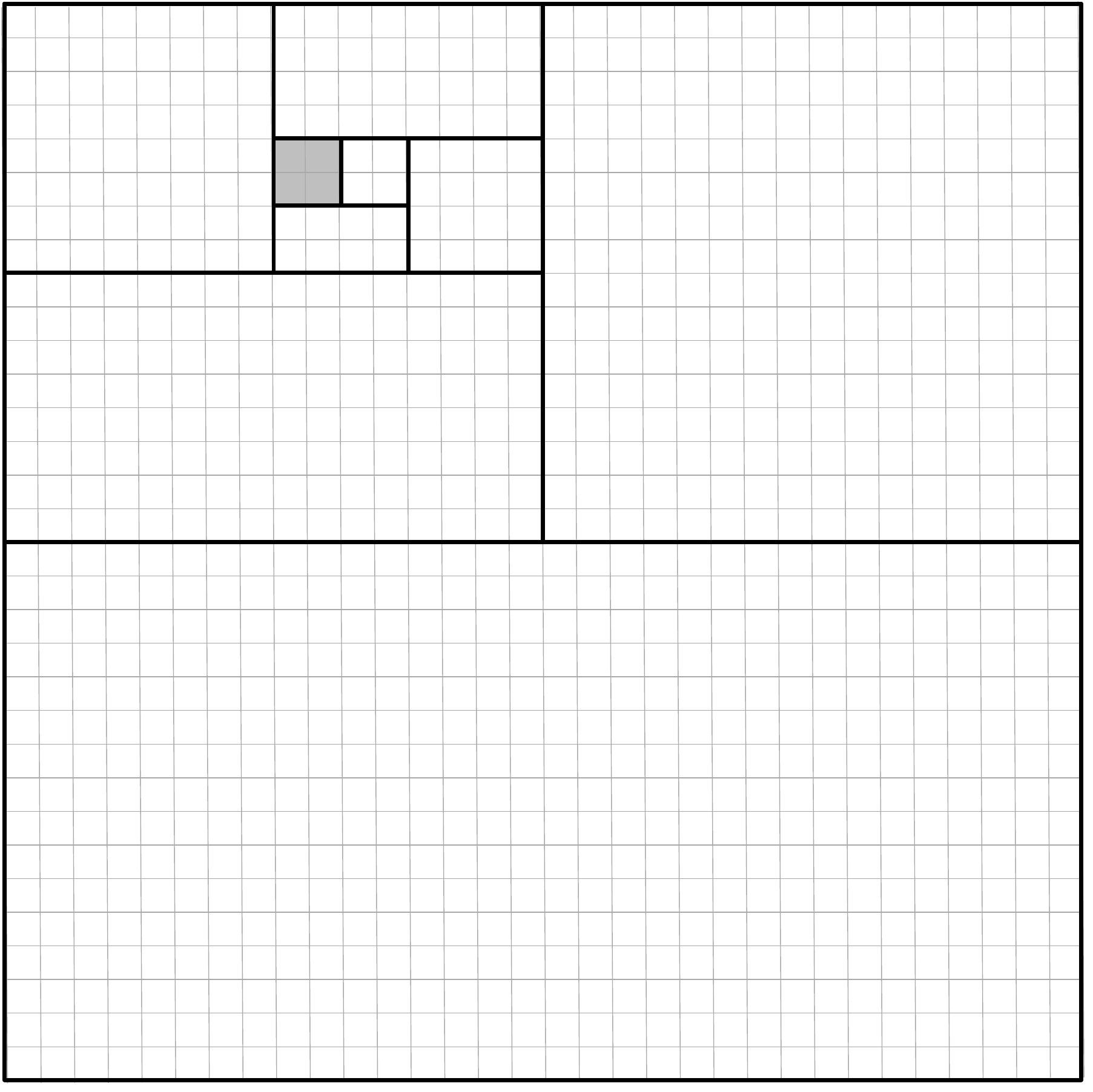}
\caption{Illustration of pieces in a recursive decomposition of the alignment graph. Diagonal edges are not show to avoid clutter. All rectangular pieces that contain the gray square form a single root-to-leaf path in the tree $\AG$.} \label{fig:grid}	
\end{center}
\end{figure}

Consider a piece $P \in \AG$.
The set $\partial P$ of \emph{boundary vertices} of a piece $P$ consists of those vertices who have neighbours that are not in $P$.
We call the vertices in $P\setminus \partial P$ the \emph{internal vertices} of $P$.
We denote by $\bot{P}$ the set of boundary vertices of $P$ that are either rightmost or bottommost in $P$, and by $\top{P}$ the set of boundary vertices of $P$ that are either leftmost or topmost in $P$.
We consider each of $\bot{P}$ and $\top{P}$ to be ordered, such that adjacent vertices are consecutive and the earliest vertex is the bottom-left one.
Therefore, whenever convenient, we refer to subsets of $\bot{P}$ and $\top{P}$ as sequences.
We define the {\em outside} of $P$, denoted by $\out{P}$, to be the set of vertices of $G \setminus (P\setminus \partial P)$ that are reachable from some vertex in $P$, i.e.~the vertices of $G$ that are not internal in $P$ and are to the right or below some vertex of $P$.
Note that any path from a vertex $u\in P$ to a vertex $v \in \out{P}$ must contain at least one vertex from $\bot{P}$, and any path from a vertex $u \not\in P$ to a vertex $u \in P$ must contain at least one vertex from $\top{P}$.

For any $r \in [1,nm]$, an $r$-{\em division} of $G$ is a decomposition of $G$ to pieces of size $\cO(r)$, each with $\cO(\sqrt{r})$ boundary vertices.
Clearly, such a decomposition can be retrieved from $\AG$, with all nodes being of the same depth.
In particular, we will use recursive $(r_t, \ldots, r_1)$-divisions, where for every $i<t$, each piece of the $r_i$-division must be contained in some piece of the $r_{i+1}$-division.
By convention, we will have $r_t$ being a single piece consisting of the entire graph~$G$.
Such a recursive division can be materialized as follows: 
First, we select the appropriate depth of $\AG$ for each $r_i$-division and mark all nodes of this depth.
Then, we contract every edge of $\AG$ both of whose endpoints are not marked.
Such a recursive $(r_t, \ldots, r_1)$-division can thus also be represented by a tree, which we will denote by $\TG$.

\begin{figure}[h]
\begin{center}
\includegraphics[page=4,scale=0.55]{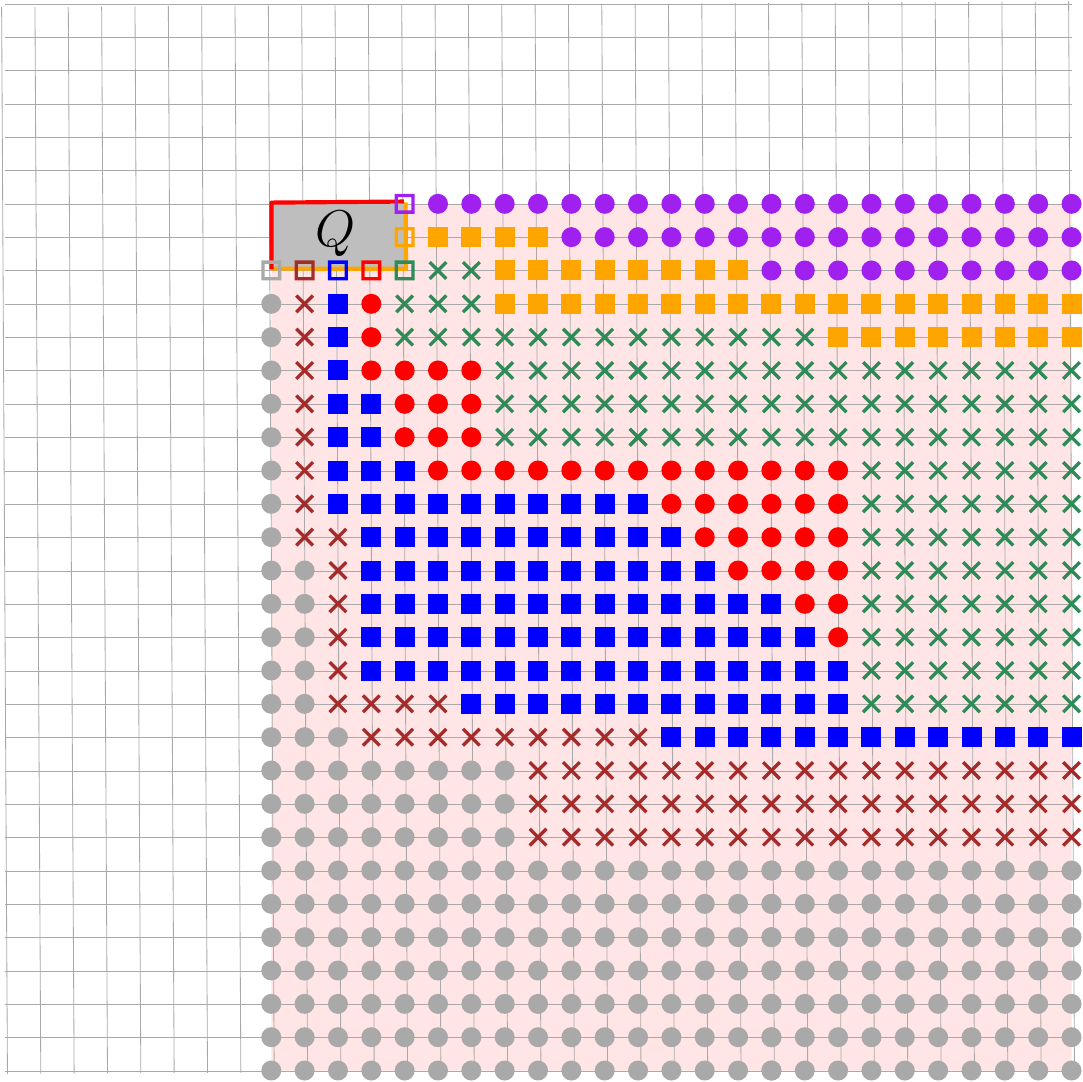}
\caption{A piece $Q$ (shaded gray). $\top Q$ is indicated by a red line, and $\bot Q$ by an orange line. $\out{Q}$ is shaded pink.
The Voronoi diagram for $\out Q$ with sites $\bot Q$ (boxes) is also illustrated. Each site has a distinct color. Vertices in each Voronoi cell are indicated by a matching color.
}\label{fig:piece} 	
\end{center}
\end{figure}

\subparagraph*{Voronoi diagrams.}
Let $H$ be a directed planar graph with real edge-lengths, and no negative-length cycles. Let $h$ be a face of $H$, 
and let $S$ be the set of vertices (called {\em sites}) of $h$. Each site $s\in S$ has a weight $\weight(s) \geq 0$ associated with it. The additively weighted distance $d^\weight(s,v)$ between a site $s \in S$ and a vertex $v \in H$ is defined as $\weight(s)$ plus the length of the shortest $s$-to-$v$ path in $H$.

The {\em additively weighted Voronoi diagram} $\VD(S, \weight)$ of $H$ is a partition of the vertices of $H$ into pairwise disjoint sets, one set $\Vor(s)$ for each site $s \in S$. The set $\Vor(s)$, called the {\em Voronoi cell} of $s$, contains all vertices of $H$ that are closer (w.r.t. $d^\weight(.,.)$) to $s$ than to any other site in $S$.
If $v \in \Vor(s)$ then we call $s$ the site of $v$, and say that $v$ belongs to the site $s$.
Throughout the paper, we will only consider additively weighted Voronoi diagrams for the outside $\out{P}$ of a piece $P \in \AG$ with sites $S \subseteq\bot{P}$. We next discuss the structure of such Voronoi diagrams.

We resolve ties between sites in favor of the site $s=(x,y)$ for which $(\weight(s),x,y)$ is lexicographically largest.
Since the alignment graph is planar, this guarantees that the vertices in $\Vor(s)$ are spanned by a subtree of a shortest paths tree rooted at $s$:
for every vertex $v \in \Vor(s)$, for any vertex $u$ on a shortest $s$-to-$v$ path, we must have $u \in \Vor(s)$.
Hence, each Voronoi cell is a simply connected region of the plane.
The structure of the alignment graph dictates that any shortest path is monotone in the sense that it only goes right and/or down. This property immediately implies the following lemma.
\begin{lemma} \label{lem:contig}
For any $a\leq c$ and $d\leq f$, if $u=(a,f)$ and $v=(c,d)$ both belong to $\Vor(s)$ then every vertex $w=(b,e)$ with $a\leq b\leq c$ and $d\leq e \leq f$ also belongs to $\Vor(s)$.
\end{lemma}
\begin{proof}
Suppose $w=(b,e)$ belongs to $\Vor(s')$ for some $s' \neq s$.
Since shortest paths only go right and down, the shortest $s'$-to-$w$ path must cross either the $s$-to-$u$ path or the $s$-to-$v$ path, which is a contradiction. 
\end{proof}

Lemma~\ref{lem:contig} together with the fact that $\Vor(s)$ is connected implies the following characterization of the structure of $\Vor(s)$, which roughly says that $\Vor(s)$ has the form of a double staircase, as illustrated in~\cref{fig:piece}.

\begin{corollary}\label{cor:structure}
For any row $a$ and any site $s$, the vertices of row $a$ that belong to $\Vor(s)$ form a contiguous interval of columns $[i_a,j_a]$. 
Furthermore, the sequences $i_a$ and $j_a$ are monotone non-decreasing and $i_a \leq i_{a+1} \leq j_{a} \leq j_{a+1}$.
\end{corollary}

\begin{corollary}\label{lem:bottomright}
There is a rightmost vertex $\last{s}$ in $\Vor(s)$ that is also a bottommost one.
\end{corollary}
 
 \begin{corollary}\label{cor:last-s-reachable}
 	For every $v \in \Vor(s)$,  $\Vor(s)$ contains a path from $v$ to $\last{s}$. 
 \end{corollary}

Our representation of a Voronoi diagram for $\out{P}$ with sites $\bot{P}$ consists of the following. For each $s \in \bot{P}$, we  store:
\begin{enumerate}
\item the rightmost bottommost vertex $\last{s}$ of $\Vor(s)$, defined in~\cref{lem:bottomright},
\item a vertex $\med{s,\last{s}}$ on the shortest $s$-to-$\last{s}$ path, whose  definition will be given later.
\end{enumerate}
\section{The Alignment Oracle}
In this section, we describe our oracle and prove that its space and query time are as in Theorem~\ref{thm:main}. In the next section we will show how to construct the oracle in $N^{1+o(1)}$ time. 

Consider a recursive $(r_{t},\ldots,r_1)$-division of $G$ for some $N=r_{t}> \cdots > r_1=\cO(1)$ to be specified later.
Recall that our convention is that the $r_t$-division consists of $G$ itself.
We also consider an $r_0$-division, in which each vertex $v$ of $G$ is a singleton piece.
Let us denote the set of pieces of the $r_i$-division by $\mathcal{R}_i$. Let $\TG$ denote the tree representing this recursive $(r_t,\ldots,r_0)$-division (each singleton piece $\{v\}$ at level $0$ is attached as the child of a piece $P$ at level $1$ such that $v \in \bot{P}$).

\begin{figure}[h]
\begin{center}
\includegraphics[page=5,scale=0.45]{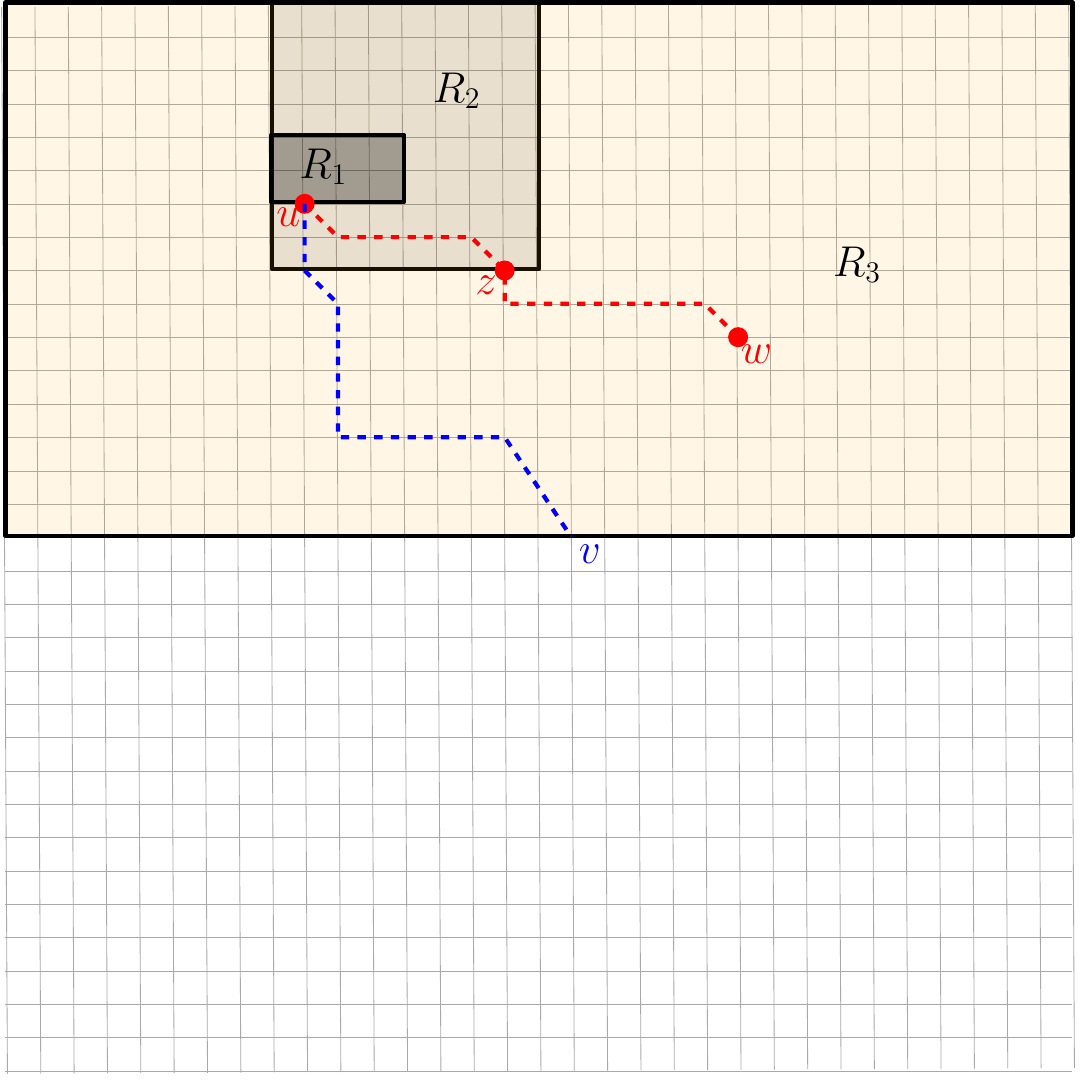}
\caption{A vertex $u$ and all the pieces of a recursive $(r_4,r_3,r_2,r_1,r_0)$-division that contain $u$. The piece $R_0$ consists of just $u$, and the piece $R_4$ is the entire alignment graph. Note that $u \in \bot{R_1}$. In this example, $\lev(u) = 1$, $\anc(u,w) = 3$, and $\anc(u,v) = 4$ (because $v \in \partial R_3$). A $u$-to-$w$ shortest path $\rho$ is shown in dashed red. $\med{u,w}$ is the last vertex of $\rho$ that belongs to $\bot {R_2}$, which is the vertex $z$.
The distance from $u$ to $w$ is $\dist(u,z) + \dist(z,w)$. 
$\dist(u,z)$ is stored in the reverse MSSP of $R_2$. $\dist(z,w)$ is stored in the MSSP of $R_3 \setminus (R_2 \setminus \partial R_2)$. Similarly, a $u$-to-$v$ shortest path $\rho$ is shown in dashed blue. Because $v \in \partial R_3$, $\med{u,v}$ is $v$ itself.
}\label{fig:recursive} 	
\end{center}
\end{figure}

The oracle consists of the following. 
For each $0 \leq i \leq t-1$, for each piece $P \in \mathcal{R}_i$ whose parent in $\TG$ is $Q \in \mathcal{R}_{i+1}$:
\begin{enumerate}
\item If $i>0$, we store an MSSP with sources $\bot{P}$ for the graph obtained from $P$ by flipping the orientation of all edges; we call this the \emph{reverse MSSP of $P$}. 
\item If $i>0$, we store an MSSP with sources $\bot{P}$ for $Q \setminus (P \setminus \partial P)$.
\item If $i<t-1$, for each vertex $u \in \bot{P}$ we store $\VD(u,Q)$: the Voronoi diagram for $\out{Q}$ with sites  $\bot{Q}$ and additive weights the distances in~$G$ from $u$ to these sites.
\end{enumerate}

To complete the description of the oracle it remains to specify the definition of $\med{s,\last{s}}$.
Before doing so, let us distinguish, for a (source) vertex $u$ and a (target) vertex~$v$, 
two levels of the recursive division that are of interest.
Let $R_0$ be the singleton piece $\{u\}$. Let $R_1, R_2, \ldots, R_t$ be the ancestors of $R_0$ in $\TG$.
Note that $u \in \bot{R_i}$ for a non-empty prefix of the sequence of ancestors $R_0, R_1, \dots, R_t$.
Let $\lev(u) = \argmax_i\{u \in \bot{R_i}\}$.
Further, let $\anc(u,v)= \argmin_i \{v \in R_i \setminus \partial R_i\}$. Note that $\anc(u,v)$ is well defined since $v \in R_t = G$ and $\partial G = \emptyset$.
Also note that if $v$ is reachable from $u$ then $\lev(u) < \anc(u,v)$. This is because all vertices in $R_{\lev(u)}\setminus \partial R_{\lev(u)}$ are unreachable from $u$. 

Denote $H=R_{\anc(u,v)-1}$.
We define $\med{u,v}$ as any boundary vertex of $\bot{H}$ that lies on a shortest $u$-to-$v$ path $\rho$.
The idea behind this definition is that $\med{u,v}$ partitions this $u$-to-$v$ path into a prefix and a suffix, each of which is represented in one of the MSSP data structures stored for $H$; The prefix of $\rho$ ending at $\med{u,v}$ is represented in the shortest path tree rooted at $\med{u,v}$ in the reverse MSSP of $H$.
The suffix of $\rho$ starting at $\med{u,v}$ is represented in the shortest path tree rooted at $\med{u,v}$ in the MSSP for $H' \setminus(H\setminus \partial H)$ with sources $\bot{H}$, where $H'=R_{\anc(u,v)}$ is the parent of $H$ in $\TG$. 
This allows us to efficiently compute $\dist(u,v)$ given $\med{u,v}$. See~\cref{fig:recursive}.
This concludes the description of the oracle.

\begin{lemma}\label{lem:space}
The oracle occupies space $\cO \left(N\log^2 N + N\log N \cdot \sum_{i=0}^{t-1} r_{i+1}/r_i\right)$.
\end{lemma}
\begin{proof}
The reverse MSSPs over all pieces of $\AG$ require $\cO(N \log^2 N)$ space, since $\sum_{P \in \AG} |P|=\cO(N \log N)$ and since the reverse MSSP of a piece $P$ requires space $\cO(|P|\log|P|)$.

For each $i \in (0,t-1]$, for each of the $\cO(N/r_i)$ pieces in $\mathcal{R}_i$, we store an MSSP of size $\cO(r_{i+1} \log r_{i+1})$.
For each $i \in [0,t-1)$, for each of the $\cO(N/r_i)$ pieces in $\mathcal{R}_i$, we store $\cO(\sqrt{r_i})$ Voronoi diagrams each of size $\cO(\sqrt{r_{i+1}})$. The stated bound follows.
\end{proof}

\subparagraph*{Query.}
We now describe how to answer a distance query $\dist(u,v)$. 
First, note that if $u$ and $v$ are in the same piece $P$ in $\mathcal{R}_{1}$, we can report $\dist(u,v)$ in $\cO(1)$ time using brute force.

Let $R_0$ be the singleton piece $\{u\}$. As before, let $R_1, R_2, \ldots , R_t$ be the ancestors of $R_0$ in~$\TG$.
Since the distance query originates from an LCS query, $v$ must be reachable from $u$.
Let $\ell=\lev(u)$ and $h=\anc(u,v)$.
We then have $u \in \bot{R_\ell}$ and $v \in \out{R_\ell}$.
We will answer $\dist(u,v)$ by identifying $\med{u,v}$, which lies on $\bot{ R_{h -1}}$.
As explained above, we can then obtain $\dist(u,v)$ from the MSSP data structures stored for~$R_{h-1}$.
See~\cref{alg:dist}.

\alglanguage{pseudocode}
\begin{algorithm}[h]
\caption{\textsc{Dist}$(u,v)$}\label{alg:dist}
\begin{algorithmic}[1]
\If{ $u$ and $v$ belong to the same piece in $\mathcal R_1$}
	\State \Return the answer by brute force
\EndIf
\State $w \leftarrow$ \textsc{GetLast}$(u,v)$
\State \Return $\dist(u,w) + \dist(w,v)$
\Statex
\end{algorithmic}
  \vspace{-0.4cm}
\end{algorithm}

We now show how to implement the procedure \textsc{GetLast}  for finding $\med{u,v}$; see~\cref{alg:getlast} for a pseudocode.
First, note that if $h=\ell+1$ we can simply return $u$ as $\med{u,v}$.
Hence, in what follows, we assume that $h>\ell+1$.
The procedure \textsc{GetLast} proceeds in iterations for $ i = \ell +1,  \dots, h - 1$. 
At the beginning of the iteration with value $i$, the procedure has a subset $W_{i-1}$ of $\bot{R_{i-1}}$ such that some $w \in W_{i-1}$ belongs to a shortest $u$-to-$v$ path. 
We initially set $W_{\ell} := \{u \}$, which trivially satisfies the requirement for $i=\ell+1$. 
For each $w \in W_{i-1}$ the iteration uses a procedure \textsc{GetNextCandidates} that adds at most two vertices of $\bot{R_{i}}$ to $W_i$. The guarantee is that, if $w$ is a vertex of $\bot{R_{i-1}}$ that belongs to a shortest $u$-to-$v$ path, 
then at least one of the two added vertices also belongs to a shortest $u$-to-$v$ path.
Since $|W_{i}|\leq 2|W_{i-1}|$, at the end of the last iteration (the one for $h - 1$), we have a subset $W_{h - 1}$ of at most $2^t$ vertices of $\bot{R_{h - 1}}$, one of which can be returned as $\med{u,v}$.
To figure out which one, for each such vertex $w$, we use the MSSP data structures to compute $\dist(u,w) + \dist(w,v)$, and return a vertex for which the minimum is attained.

\begin{algorithm}[h]
\caption{\textsc{GetLast}$(u,v)$}\label{alg:getlast}
\begin{algorithmic}[1]
\State $\ell \leftarrow \lev(u)$
\State $h \leftarrow \anc(u,v)$
\State $W_\ell \leftarrow \{u\}$
\For{$  i = \ell+1$ to $h-1$}
	\State $W_i \leftarrow \emptyset$ 
	\For {each $w \in W_{i-1}$}
		\State $W_{i} \leftarrow W_i \ \cup $ \textsc{GetNextCandidates}$(w,i,v)$
	\EndFor
\EndFor
\State \Return $\argmin_{w\in W_{h-1}} \dist(u,w) + \dist(w,v)$
\Statex
\end{algorithmic}
  \vspace{-0.4cm}
\end{algorithm}

It remains to describe the procedure \textsc{GetNextCandidates}. 
Consider any $w \in W_{i-1}$. 
To reduce clutter, let us denote $R_i$ by $Q$. 
Since $w \in \bot{R_{i-1}}$, the Voronoi diagram $\VD(w,Q)$ for $\out Q$ is stored by the oracle. The procedure \textsc{GetNextCandidates} finds two sites of $\VD(w,Q)$, one of which is the site of $v$. 
Indeed, if $w$ is a vertex on a $u$-to-$v$ shortest path then the site of $v$ in $\VD(w,Q)$ is a vertex of $\bot Q$ on a shortest $u$-to-$v$ path.

Let the top-left and bottom-right vertices of $Q$ be $(x_\tp,y_\tp)$ and $(x_\bt,y_\bt)$, respectively.
Every vertex $v=(x_v,y_v)$ of $\out{Q} \setminus \partial Q$ either has $x_v> x_\bt$ or $y_v > y_\bt$.
We describe the case when $x_v>x_\bt$; the case when $y_v > y_\bt$ is analogous.
Let $\Gamma$ denote the set of $s \to \last{s}$ paths $\rho_s$ stored in $\VD(w,Q)$ according to the order of the sites $s$ along $\bot{Q}$.
For every row $x>x_\bt$, let $\Gamma_x$ denote 
the set of paths $\rho_s$ stored in $\VD(w,Q)$ such that $\rho_s$ intersects row $x$, ordered according to the order of the intersection vertices along row $x$.

\begin{lemma}\label{lem:gamma-inclusion}
For every $x, x'$ with $x>x'$, $\Gamma_x$ is a subsequence of $\Gamma_{x'}$.
\end{lemma}
\begin{proof}
This is a direct consequence of the fact that each path $\rho_s$ goes monotonically down and right, and from the fact that $\rho_s$ and $\rho_{s'}$ are disjoint for $s \neq s'$.
\end{proof}

We define the set of {\em critical rows} to be all rows $x$ such that $(x,y) = \last{s}$ for some $y\in [0,n]$ and $s \in \bot{Q}$.
Lemma~\ref{lem:gamma-inclusion} implies that if we consider the evolution of the  sequences $\Gamma_x$ as $x$ increases from $x_\bt$ to $m$ as a dynamic process, changes occur only at critical rows. More precisely, if the row of $\last{s}$ is $x$ for some site $s$, then $\rho_s \in \Gamma_x$ but $\rho_s \notin \Gamma_{x+1}$.
We can therefore maintain the sequences $\Gamma_x$ 
in a persistent binary search tree.
(A binary search tree can be made partially persistent at no extra asymptotic cost in the update and search times using a general technique for pointer-machine data structures of bounded degree~\cite{DBLP:journals/njc/Brodal96}.)
Initially, the BST stores the sequence $\Gamma_{x_\bt +1}$. Then, we go over the critical rows in increasing order, and remove the path $\rho_s$ from the BST when we reach the row of $\last{s}$.

For any $x_\bt < x \leq n$, we can access the BST representation of $\Gamma_x$ by finding the predecessor $x'$ of $x$ among the critical rows, and accessing the persistent BST at time $x'$.

Let $v = (x,y)$. 
We say that $v$ is right (left) of a path $\rho_s \in \Gamma_x$ if $x$ is greater (smaller) than any vertex of $\rho_s$ at row $x$.
We will either find a path $\rho_s \in \Gamma_x$ to which $v$ belongs, 
or identify the last path $\rho_s \in \Gamma_x$ such that $v$ is right of $\rho_s$. 
In the former case the site of $v$ is $s$,  and in the latter case the site of $v$ is either $s$ or the successor of $s$ in $\Gamma_x$.

Recall that (1) the MSSP data structure, given a root vertex $r$ and two vertices $w,z$, can determine in $\cO(\log N)$ time whether $w$ is left/right/ancestor/descendant of $z$ in the shortest path tree rooted at $r$, (2) for each shortest path $\rho_s$ represented in $\VD(w,Q)$, the representation contains $\med{s,\last{s}}$, and   
(3) the prefix of $\rho_s$ ending at $\med{s,\last{s}}$ is represented in the shortest path tree rooted at $\med{s,\last{s}}$ in the reverse MSSP of $H$.
Similarly, the suffix of $\rho_s$ starting at $\med{s,\last{s}}$ is represented in the shortest path tree rooted at $\med{s,\last{s}}$ in the MSSP for $H' \setminus(H\setminus \partial H)$, where $H'$ is the parent of $H$.
 
We perform binary search on $\Gamma_x$ to identify the path $\rho_s$ such that either $v\in \rho_s$ or $\rho_s$ is the last path of $\Gamma_x$ that is left of $v$. 
Focus on a step of the binary search that considers a path~$\rho_s$. 
Denote $\med{s,\last{s}}=(x_b,y_b)$.
If $y<y_b$, we query the MSSP structure that contains the prefix of $\rho_s$, and otherwise we query the MSSP data structure that contains the suffix of $\rho_s$. 
In either case, the query either returns that $v$ is on $\rho_s$ or tells us whether $v$ is left or right of $\rho_s$.
In the former case we conclude that the site of $v$ is $s$. In the latter case we continue the binary search accordingly. 
Each step of the binary search takes $\cO(\log n)$ time.
Note that $\log n = \cO(\log N)$.
Thus, the binary search takes $\cO(\log^2 N)$ time, and when it terminates we have a site $s$ that is either the site of $v$ or the site such that $\rho_s$ is the last path of $\Gamma_x$ that is left of $v$.
This implies that the site of $v$ is either $s$ or the successor of $s$ in $\Gamma_x$, and concludes the description of \textsc{GetNextCandidates}.

\begin{lemma}\label{lem:query}
The oracle answers distance queries in time $\cO(2^t \log^2 N)$.
\end{lemma}
\begin{proof}
First, $\lev(u)$ and $\anc(u,v)$ can be (naively) computed in $\cO(\log N)$ time by going over the ancestors of ${u}$ in $\mathcal{T}$: the coordinates of the corners of each (rectangular) ancestor piece of $u$ can be retrieved in $\cO(1)$ time per ancestor, and for each of them we can check in $\cO(1)$ time whether $v \in R\setminus \partial R$, using $v$'s coordinates.
Overall, for a $\dist(u,v)$ query, we make $\cO(2^t)$ calls to \textsc{GetNextCandidates}, each requiring $\cO(\log^2 N)$ time, for a total of $\cO(2^t \log^2 n)$ time.
Finally, we make $\cO(2^t)$ queries to the MSSP data structures, requiring $\cO(2^t \log N)$ time in total.
\end{proof}

\begin{remark}
Since our query procedure computes $\med{u,v}$, and we have MSSP data structures that capture the $u$-to-$\med{u,v}$ and the $\med{u,v}$-to-$v$ shortest paths, an optimal alignment can be returned in time proportional to the total length of the two substrings.
\end{remark}

By setting the $r_i$'s appropriately, we obtain the following tradeoffs, which are identical to those of Pettie and Long for arbitrary planar graphs~\cite{DBLP:journals/corr/abs-2007-08585}.

\begin{proposition}\label{prop}
For two strings of lengths $m$ and $n$, with $N=mn$, there is an alignment oracle achieving either of the following tradeoffs:
\begin{itemize}
\item $N\log^{2+o(1)} N$ space and $N^{o(1)}$ query time,
\item $N^{1+o(1)}$ space and $\log^{2+o(1)} N$ query time.
\end{itemize}
\end{proposition}
\begin{proof}
The space of the oracle is $\cO \left(N\log^2 N + N\log N \cdot \sum_{i=0}^{t-1} r_{i+1}/r_i\right)$ by~\cref{lem:space}.
We will choose $r_i$'s for $i\geq 1$ to be a geometric progression with common ratio $p$ to be specified below.
In that case, $t=\cO(\log_p N)$ and the space becomes $\cO(N\log^2 N + N\log N \cdot p \log_p N)$.
First, let us set $p=N^{1/g(N)}$ for some $g(N)$ which is $\omega(\log N/ \log\log N)$ and $o(\log N)$. Then, $p=2^{\log N/ g(N)}=2^{o(\log\log N)}=\log^{o(1)} N$.
We get $\cO(N \log N\cdot 2^{\log N/g(N)} \log N)=N\log^{2+o(1)} N$ space and $N^{o(1)}$ query time.
Second, let us set $p=N^{1/f(N)}$, for some $f(N)$ which is $\omega(1)$ and $o(\log\log N)$. 
We get $N^{1+o(1)}$ space and $\log^{2+o(1)} N$ query time.
\end{proof}

The following observation will prove useful in the efficient construction algorithm of the oracle that will be presented in the next section.

\begin{observation}\label{obs:limited}
The query algorithm for $\dist(u,v)$  takes $\cO(2^{t-\lev(u)}\log^2 n)$ time and uses only Voronoi diagrams $\VD(u,Q)$ for $Q \in \mathcal{R}_i$ with $i > \lev(u)$. 
\end{observation}

\section{An Efficient Construction Algorithm}
In this section, we present an algorithm  for constructing the alignment oracle in $N^{1+o(1)}$ time (thus completing the proof of Theorem~\ref{thm:main}). 
The computation of the recursive decomposition, the recursive $(r_t, \ldots, r_0)$-division and all of the MSSP structures stored for all pieces in $\AG$ can be done in $\cO(N\log^2 N)$ time. It therefore only remains to analyze the time it takes to construct all the representations of Voronoi diagrams stored by the oracle. 

Consider some additively weighted Voronoi diagram for $\out{Q}$ with sites a subsequence $U$ of $\bot{Q}$---we will only build Voronoi diagrams with $U=\bot{Q}$, but during the analysis, we will also consider Voronoi diagrams with sites $U \subseteq \bot{Q}$.
In what follows, when we talk about a piece $H \neq Q$, we will really mean its intersection with $\out{Q}$, assuming that it is non-empty.
Similarly, when we talk about $\partial{H}$, $\top{H}$, and $\bot{H}$ we will really mean the intersection of $\partial{H}$, $\top{H}$, and $\bot{H}$ with $\out{Q}$, respectively. 
See~\cref{fig:cover} for an illustration.

\begin{figure}[h]
\begin{center}
\includegraphics[page=6,scale=0.45]{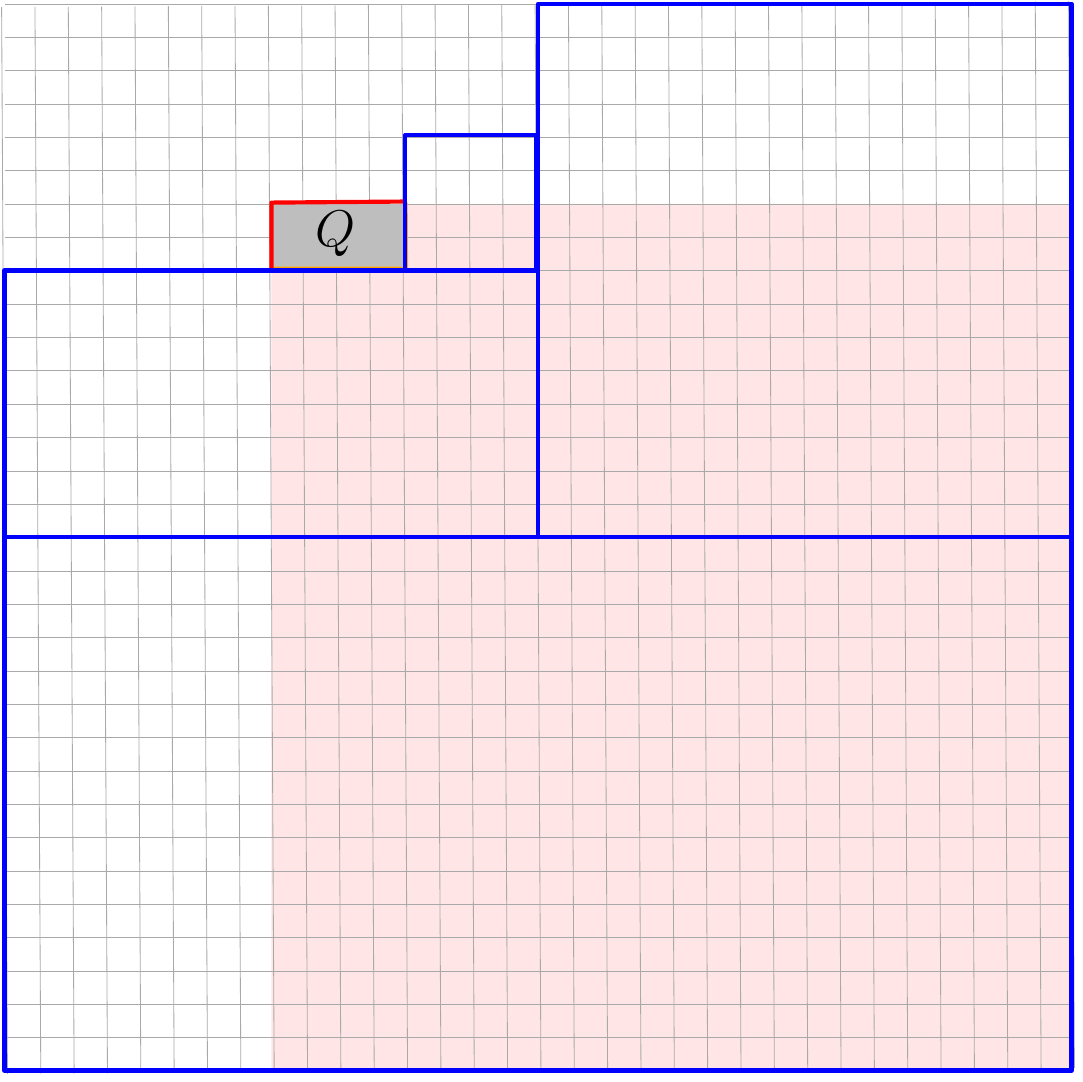}
\caption{Covering $\out Q$ (shaded pink) with siblings of ancestors of $Q$ in $\AG$ (blue boxes). When we refer to a piece $H$ (e.g. a blue box), we only refer to the portion of $H$ that belongs to $\out Q$.
}\label{fig:cover} 	
\end{center}
\end{figure}

\begin{lemma}\label{lem:cormon}
Let $X \in \{\top{H},\bot{H}\}$.
Let $u,v\in X$ belong to distinct Voronoi cells.
If $u$ precedes $v$ (in $X$) then the site $s_u \in U$ of $u$ precedes (in $U$) the site $s_v \in U$ of $v$.
\end{lemma}
\begin{proof}
For any vertex $a \in X$, that belongs to the cell of a site $s_a$, all 
vertices in the shortest $s_a$-to-$a$ path belong in $\Vor(s_a)$.
Towards a contradiction suppose that $s_u$ succeeds $s_v$ in $U$. 
By the planarity of the graph and the fact that paths can go only down and right it follows that the shortest $s_u$-to-$u$
path must cross the shortest $s_v$-to-$v$ path in some vertex $b$.
Then $b \in \Vor(s_u) \cap \Vor(s_v)$, which is a contradiction as Voronoi cells are disjoint.
\end{proof}

The above lemma means that the vertices of each of $\top{H}$ and $\bot{H}$ can be partitioned into maximal contiguous intervals of vertices belonging to the Voronoi cell of the same site in $U$.
When we say that we compute the partition of $\top{H}$ or $\bot{H}$ with respect to $U$, we mean that we compute the subsequence of sites in $U$ of these intervals by identifying the endpoints of each interval and the site of each interval.

The following simple observation allows us to compute partitions using binary search.
It says that a piece $H$ contains $\last{s}$ if and only if 
$s$ is the site of some vertex in $\top{H}$, and $s$ is not the
site of any vertex in $\bot{H}$.

\begin{lemma}\label{lem:zoom}
For any $s\in S$ and any level $\ell$, there is a unique level-$\ell$ piece $H\in \AG$ for which $\Vor(s) \bigcap \top{H}\neq \emptyset$ and $\Vor(s) \bigcap \bot{H}= \emptyset$, and this piece contains $\last{s}$.
\end{lemma}
\begin{proof}
By~\cref{cor:last-s-reachable}, for all $v \in \Vor(s)$, there exists a $v$-to-$\last{s}$ path all of whose vertices are in $\Vor(s)$.
Hence, for every level-$\ell$ piece $H$ for which $\Vor(s) \bigcap \top{H}\neq \emptyset$ and $\last{s} \not\in H$, we must also have that $\Vor(s) \bigcap \bot{H} \neq \emptyset$.
We can thus focus on the at most four level-$\ell$ pieces that contain $\last{s}$.
It is readily verified that the bottom-right of those pieces is the only one for which $\Vor(s) \bigcap \bot{H}= \emptyset$.
\end{proof}

\begin{remark}
The condition in the statement of Lemma~\ref{lem:zoom} is equivalent to: $\last{s} \in H \setminus \bot{H}$.
\end{remark}

Lemma~\ref{lem:zoom} provides a criterion on the partitions of $\top H$ and $\bot H$ for determining whether a piece $H$ contains any vertex $\last{s}$. 
The following lemma describes a binary search procedure, 
 {\sc Partition}, which gets as input a sequence $U$ of candidate sites, and returns the partition of 
$\top{H}$ or $\bot{H}$; i.e., the subsequence of sites in $U$ whose Voronoi cells contain the vertices of $\top{H}$ or $\bot{H}$. The procedure {\sc Partition} will be a key element in the overall construction algorithm. 
The following lemma describes an implementation of {\sc Partition} using distance queries $\dist(u, v)$ with $u \in \bot{Q}$ and $v \in \partial H$. We will ensure that such queries can be answered efficiently whenever {\sc Partition} is called by the main algorithm.

\begin{lemma}\label{lem:partition}
Given a sequence $U \subseteq \bot{Q}$ of sites and their additive weights, we can perform the procedure \textsc{Partition} (that computes a partition of $\top{H}$ or $\bot{H}$ w.r.t.~$U$) in the time required by $\cO(|U|\cdot \log n)$ distance queries $\dist(u, v)$ with $u \in \bot{Q}$ and $v \in \partial H$.
\end{lemma}
\begin{proof}
We will only prove the statement for $\top{H}$ as the case of $\bot H$ is analogous.
We start with a single interval, which is all of $\top H$.
We will call an interval \emph{active} if we have not concluded that all of its vertices belong to the same Voronoi cell.
For each active interval $L$, we have a set $C_L$ of \emph{candidate sites}.
Thus, initially, the single interval $\top{H}$ is active, and $U$ is the set of its candidate sites.

The algorithm proceeds by divide and conquer. 
As long as we have an active interval $L$, we perform the following:
we compute the site $u \in C_L$ with the minimum additively weighted distance to the midpoint of $L$. This is done in the time required by $C_L$ distance queries of the form specified in the statement of the lemma.
Then, we split $L$ at this midpoint: for the left part of $L$ the set of candidate sites is now $\{v \in C_L : v \leq u\}$,
while for the right part of $L$ the set of candidate sites is now $\{v \in C_L: v\geq u\}$. If either of these two sets is of size $1$, the corresponding interval becomes inactivate.
That is, we recurse on at most two active intervals of roughly half the length.
In the end, in a left-to-right pass, we merge consecutive intervals all of whose vertices belong to the same Voronoi cell.

Let us now analyze the time complexity of the above algorithm.
First, observe that the sequences of candidates of any two intervals at the same level of the recursion are internally disjoint. Thus, each site is a candidate for at most two active intervals at the same recursive level.
Second, at level $j$ of the recursion, the length of every active interval is $\cO(|\top{H}|/2^j)$.
Hence, the total time required to process all intervals is proportional to the time required by $\cO( |U| \cdot \log n)$ distance queries $\dist(u, v)$, with $u \in \bot{Q}$ and $v \in \partial H$.
\end{proof}

We now present the algorithm for computing the representations of the Voronoi diagrams stored by the oracle.
The algorithm performs the computation in order of decreasing levels of the recursive $(r_t,\ldots,r_0)$-division.

Consider some level $i$, and assume that we have already computed all Voronoi diagrams $\VD(u,R)$ for pieces $R \in \bigcup_{j>i} \mathcal{R}_j$. 
Consider any piece $P\in \mathcal R_{i-1}$. 
Let $Q \in \mathcal R_{i}$ be the parent of $P$ in $\TG$. 
Our goal is to compute, for every $u \in \bot P$, the representation of $\VD(u,Q)$, the Voronoi diagram of $\out{Q}$ with sites $\bot Q$ and additive weights $\dist(u, \bot{Q})$.
Recall that this representation consists of the vertices $\last{s}$ and $\med{s,\last{s}}$ for every site $s \in \bot{Q}$. 
We would like to compute this representation in time roughly proportional to its size $|\bot Q|$.
By~\cref{obs:limited}, using the already computed parts of the oracle for levels $j>i$,
we can already answer any
distance query $\dist(s,v)$ for any $s\in \bot Q$ 
and any $v \in \out{Q}$ in $\cO(2^t \log^2 n)$ time.
These are precisely the distance queries required for computing
partitions of pieces $H$ in $\out{Q}$ w.r.t. sites in $\bot Q$ (\cref{lem:partition}).
 
The computation is done separately for each $u \in \bot{P}$. 
First, we compute the additive weights $\dist(u, \bot{Q})$ in $\cO(|\bot{Q}|\cdot \log n)$ time using the MSSP data structure stored for $Q \setminus (P \setminus \partial P)$ with sites $\bot{P}$.
Next, we cover $\out{Q}$ using $\cO(\log N)$ pieces from $\AG$ that are internally disjoint from $Q$ (i.e.~they may only share boundary vertices). 
These pieces are the  $\cO(\log N)$  siblings of the (weak) ancestors of $Q$ in $\AG$ that have a non-empty intersection with $\out{Q}$ (see~\cref{fig:cover}). 
Notice that these pieces are in $\AG$ but not necessarily in~$\TG$.

We shall find the vertices $\last{s}$ of $\VD(u,Q$) in each such piece $H$ separately. 
We invoke {\sc Partition} on $\top H$ and on $\bot H$, and use~\cref{lem:zoom} to determine whether $H$ contains any vertices $\last{s}$.
If so, we zoom in on each of the two child pieces of $H$ in $\AG$ until, after $\cO(\log N)$ steps, we get to a constant-size piece, in which we can find $\last{s}$ by brute force.
Note, however, that we are aiming for a running time that is roughly proportional to $|\bot Q|$, 
but that the running time of {\sc Partition} depends on the number of sites $U$   w.r.t.~which we partition. 
This is problematic since, e.g., when $H$ contains $\last{s}$ just for a single site $s$, we can only afford to invest $\cOtilde(1)$ time in locating $\last{s}$ in $H$. In this case, computing the partition w.r.t.~$|\bot{Q}|$ is too expensive. 
Even computing the partition just w.r.t.~the sites whose Voronoi cell has non-empty intersection with $H$, which is bounded by $|\top{H}|$, is too expensive.
To overcome this problem we will show that it suffices to compute the partition w.r.t. a smaller sequence of sites, whose size is proportional to the number of sites $s$ with $\last{s}$ in $H$ (actually in $H \setminus \bot{H}$), rather than to the size of $H$ or of $\top H$.
We call such a sequence a {\em safe} sequence of sites for $H$, which we now define formally.
Recall that the Voronoi diagram $\VD(u,Q)$ of $\out{Q}$ has sites $\bot{Q}$.
Let $U$ be a subsequence of $\bot{Q}$. 
Consider the Voronoi diagram $\VD'$ of $\out{Q}$ whose sites are just the sites of $U$ (with the same additive distances as in $\VD(u,Q)$).
We say that $U$ is \emph{safe for $H$} if and only if the sets $\{(s,\last{s}): s \mbox{ is a site and }\last{s} \in H \setminus \bot{H}\}$ are identical for $\VD'$ and $\VD(u,Q)$.

\begin{observation}\label{fact:child}
A sequence that is safe for $H$ is also safe for any child $H'$ of $H$ in $\AG$. 
\end{observation}

We will discuss the details of safe sequences after first providing the pseudocode of the procedure {\sc Zoom} for finding the vertices $\last{s}$ in a piece $H$. 

\alglanguage{pseudocode}
\begin{algorithm}[h]
\caption{\textsc{Zoom}$(U,\weight,H)$}\label{alg:zoom}
 \textbf{Input:} The additive weight $\weight(s)$ for each $s \in \bot{Q}$, a piece $H$ in $\AG$ that is internally disjoint from $Q$ and has a non-empty intersection with $\out{Q}$, and a sequence $U \subseteq \bot{Q}$ that is safe for $H$.\\
 \textbf{Output:} All vertices $\last{s}$ for $s \in \bot{Q}$ that belong to $H$. 
\begin{algorithmic}[1]
\If{$|H| = 4$}
	\State \noindent Find all vertices $\last{s}$ in $H$ for all $s \in U$ by computing $\weight(s) + \dist(s,v)$ for all $s \in U$ and all $v$ adjacent to some vertex of $H$.
\EndIf
\State $W \leftarrow $ \textsc{Partition}$(U,\weight,\top{H})$\label{zoom:line4}
\State $W' \leftarrow $ \textsc{Partition}$(U,\weight,\bot{H})$\label{zoom:line5}
\If{$W \setminus W' \neq \emptyset$} \label{zoom:beforeloop}
	\For {each child $H'$ of $H$ in $\AG$}
    	\State $Z \leftarrow$ \textsc{Partition}$(U,\weight,\top{H'})$\label{zoom:touch}
		\State $L \leftarrow Z \bigcap (W \setminus W')$ \label{zoom:last}
		\State $V \leftarrow Z \setminus \{ z \in Z \setminus L : \mbox{ both the predecessor and successor of $z$ in $Z$ are not in $L$} \}$ \label{zoom:neighb}
    	\State \textsc{Zoom}$(V,\weight,H')$\label{zoom:line8}
    \EndFor    
\EndIf
\Statex
\end{algorithmic}
  \vspace{-0.4cm}
\end{algorithm}

The procedure \textsc{Zoom} takes as input a piece $H$ and a sequence $U \subseteq \bot{Q}$ that is safe for $H$.
In order to compute $\VD(u,Q)$, we call the procedure \textsc{Zoom}$(\bot{Q},\weight,H)$ for each of the $\cO(\log N)$ pieces $H$ that we use to cover $\out{Q}$.
Clearly, in each of those initial $\cO(\log N)$ calls $\bot{Q}$ is a safe sequence of the respective piece.
In lines~\ref{zoom:line4}-\ref{zoom:beforeloop} we check whether the condition of~\cref{lem:zoom} is satisfied.
Since $U$ is a safe sequence for $H$, this allows us to infer the set $W \setminus W'=\{s \in \bot{Q} : \last{s} \in H\setminus \partial H \}$.
If $W \setminus W'$ is non-empty, we recurse on both children of $H$ in $\AG$.
Before doing so, we construct a safe sequence for each of those children, of size proportional to $|W \setminus W'|$.
In order to prove the correctness of procedure \textsc{Zoom}, it remains to show that lines~\ref{zoom:line4}-\ref{zoom:neighb} 
indeed produce such a set. The following two lemmas show this (see also~\cref{fig:safe}).

\begin{lemma}\label{lem:remove}
Let $U$ be safe for a piece $H$, such that \textsc{Partition}$(U,\weight,\top{H})=U$. 
Suppose that there are three elements $u_1,u_2,u_3$ of $U$ that appear consecutively (in this order) in both \textsc{Partition}$(U,\weight,\top{H})$ and \textsc{Partition}$(U,\weight,\bot{H})$.
Then, $U\setminus \{u_2\}$ is also safe for $H$.
\end{lemma}
\begin{proof}
To avoid confusion we denote the Voronoi diagram of $\out{Q}$ with sites $U$ by $\VD$ and the one with sites $U\setminus \{u_2\}$ by $\VD'$. We denote the Voronoi cells of $\VD$ by $\Vor(\cdot)$, and those of $\VD'$ by $\Vor'(\cdot)$.
Note that for every $u \in U \setminus \{u_2\}$, $\Vor(u) \subseteq \Vor'(u)$.
By~\cref{lem:zoom}, in $\VD$, $\last{{u_1}},\last{{u_2}},\last{{u_3}} \notin H \setminus \bot{H}$. Hence, in $\VD'$, $\last{{u_1}},\last{{u_3}} \notin H \setminus \bot{H}$.

\begin{figure}[h]
\begin{center}
\includegraphics[page=7,scale=0.6]{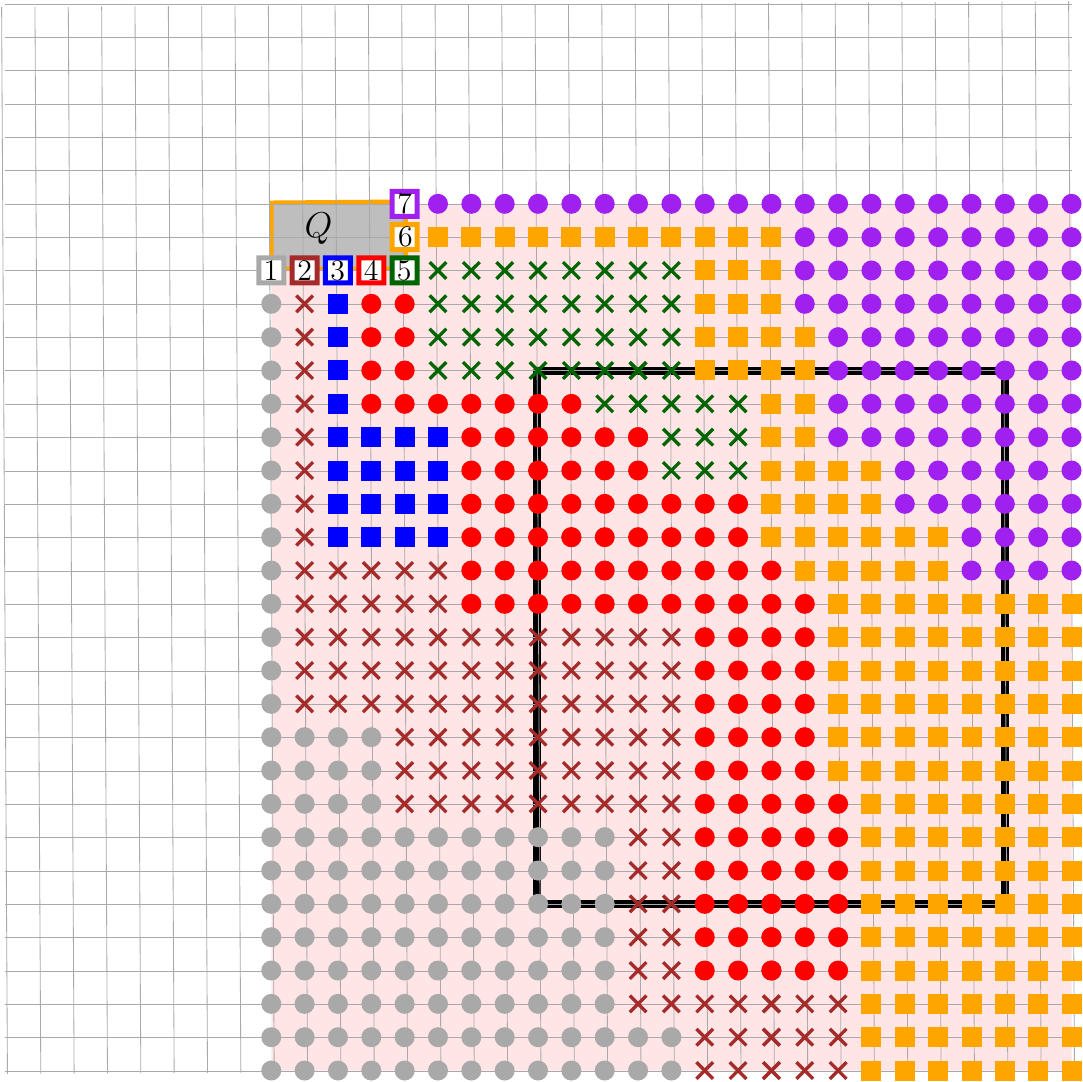}
\caption{Illustration for~\cref{lem:remove}. Part of $\out{Q}$ (pink) for some piece $Q$ (gray) is shown. A piece $H$ is indicated by a black rectangle. The sites of $\bot Q$ are numbered 1 through 7. The partition of $\top H$ w.r.t.~$\bot Q$ is $U=(1,2,4,5,6,7)$. Hence, $U$ is safe for $H$. The partition of $\bot H$ w.r.t.~$\bot Q$ is $1,2,4,6,7$. Since sites $1,2,4$ are consecutive in both partitions, $1,4,5,6,7$ is also safe for $H$.
Further, $\bot{Q}$ is also clearly safe for $H$.
However, $1,3,4,5,6,7$ may not be safe for $H$, as $\last{3}$ could be in $H \setminus \partial H$ in the Voronoi diagram with sites $1,3,4,5,6,7$ and the same additive weights.
}\label{fig:safe} 	
\end{center}
\end{figure}

\begin{claim*}
Every vertex $y$ of $\Vor(u_2) \cap H$ 
belongs either to $\Vor'(u_1)$ or to $\Vor'(u_3)$.
\end{claim*}
\begin{claimproof}
Consider the last vertex $z_1$ of $\bot{H}$ that is in $\Vor(u_1)$, and the first vertex $z_3$ of $\bot{H}$ that is in $\Vor(u_3)$. Let $\rho_1$ be a shortest $u_1$-to-$z_1$ path, and $\rho_3$ be a shortest $u_3$-to-$z_3$ path.
Note that all vertices of $\rho_1$ belong to $\Vor(u_1)$ and all vertices of $\rho_3$ belong to $\Vor(u_3)$.
Consider any vertex $y$ of $\Vor(u_2) \cap H$.
The vertex $y$ lies to the right of $\rho_1$  and to the left of $\rho_3$. In $\VD'$, the vertices of $\rho_1$ belong to $\Vor'(u_1)$ and the vertices of $\rho_3$ belong to $\Vor'(u_3)$.
Hence, by~\cref{lem:cormon}, in $\VD'$, $y$ can only belong to a site $s\neq u_2$ that is weakly between $u_1$ and $u_3$. 
Since $u_1,u_2,u_3$ appear consecutively in \textsc{Partition}$(U,\weight,\top{H})=U$, the only such sites are $u_1$ and $u_3$, and the claim follows.
\end{claimproof}

By the above claim, the sets $\{(s,\last{s}): s \mbox{ is a site and }\last{s} \in H \setminus \bot{H}\}$ are identical for $\VD$ and $\VD'$. Since $U$ is safe for $H$, so is $U\setminus \{u_2\}$.
\end{proof}

\begin{lemma}\label{lem:safe}
Suppose that $U$ is safe for $H$.
Then, in each recursive call \textsc{Zoom}$(V,\weight,H')$ made by procedure \textsc{Zoom}$(U,\weight,H)$ for a child $H'$ of $H$ in $\AG$,
$V$ is a safe sequence for $H'$.
\end{lemma}
\begin{proof}
Since $U$ is a safe sequence for $H$,  
by~\cref{lem:zoom}, after line~\ref{zoom:beforeloop}, $W \setminus W'$ is the set of sites $s$ such that $\last{s} \in H \setminus \bot{H}$. 

In line~\ref{zoom:touch}, $Z$ is the set of sites whose Voronoi cells have non-empty intersection with $H'$. This is because, by~\cref{fact:child}, $U$ is also a safe sequence for $H'$ and so $Z$ is safe for $H'$ as well.
Then, the set $L$, defined as $Z \cap (W \setminus W')$ in line~\ref{zoom:last}, is the set of sites $s$ for which $\last{s} \in H'$.
By~\cref{lem:zoom}, any site $z \in Z \setminus L$ appears both in \textsc{Partition}$(Z,\weight,\top{H'})$ and in
\textsc{Partition}$(Z,\weight,\bot{H'})$.  
In line~\ref{zoom:neighb}, we remove from $Z$ all vertices of $Z \setminus L$ that are not preceded or succeeded by a vertex in $L$. 
Therefore, by considering the removal of these sites one at a time, and directly applying~\cref{lem:remove} to each such removal,
we conclude that the resulting sequence $V$ is safe for $H'$.
\end{proof}

This establishes the correctness of our construction algorithm.
Let us now analyze its time complexity.
Initially, we make $\cO(\log N)$ calls to \textsc{Zoom}$(U,\weight,H)$, each with $U=\bot{Q}$.
In each recursive call, for a child $H'$ of a piece $H$, the set $U$ of sites is of size proportional to the size of the set $\{s \in \bot{Q} : \last{s} \in H\setminus \bot{H}\}$. 
Note that in each level of the tree $\AG$ each $s\in \bot{Q}$ is an element of exactly one such set. 
Hence, each $s \in \bot{Q}$ contributes to $\cO(\log N)$ calls to \textsc{Zoom}: the $\cO(\log N)$ initial ones and at most two more per level of $\AG$ (which is of depth $\cO(\log N)$).

Thus, by~\cref{lem:partition}, computing $\last{s}$ for all $s\in \bot{Q}$ reduces to $\cO(|\bot{Q}|\cdot \log^2 N)$ distance queries $\dist(u, v)$, with $u \in \bot{Q}$ and $v \in \out{Q}$.
We can answer each such query with the portion of the oracle that has already been computed in $\cO(2^t \log^2 N)$ time.
Now, recall that a $\dist(s,\last{s})$ query also computes $\med{s,\last{s}}$, and hence these values can also be retrieved in $\cO(2^t \log^2 N)$ time. 
Thus, $\VD(u,Q)$, which is of size $\cO(|\bot{Q}|)$ can be computed in time $\cO(|\bot{Q}|\cdot 2^t \log^4 N)$, which is $\cO(|\bot{Q}|\cdot N^{o(1)})$ for both choices of $t$ in~\cref{prop}. Therefore, the time to compute $\VD(u,Q)$ for all pieces  is $N^{o(1)} \cdot \sum_Q \cO(|\bot{Q}|) = \sum_{i=0}^{t-1} \frac{N}{r_i}\sqrt{r_i}$, which is $\cO(N^{1+o(1)})$ for both choices of $t$. 
This concludes the proof of Theorem~\ref{thm:main}. 

\section{Tradeoffs with \boldmath$o(N)$ space}
\label{sec:subquadratic}

In this section we prove~\cref{thm:frbased}. 
Recall that for this result we consider integer alignment weights upper-bounded by $w$: A weight 
$w_{match}$ for aligning a pair of matching letters, $w_{mis}$ for aligning a pair of mismatching letters, and $w_{del}$ for letters that are not aligned.
One may assume without loss of generality that $2w_{match}>2w_{mis} \geq w_{del}$~\cite{DBLP:journals/corr/abs-0707-3619}.
Given $w_{match}$, $w_{mis}$ and $w_{del}$, we define $w'_{match}=0$, $w'_{mis}=w_{match}-w_{mis}$ and $w'_{del}=\frac12 w_{match}-w_{del}$.
These weights are also upper-bounded by $w$.
Then, a shortest path (of length $W$) in the alignment grid with respect to the new weights, corresponds to a highest scoring path with respect to the original weights (of score $\frac12(m+n)w_{match}-W$). 

\subparagraph{FR-Dijkstra.} 
We define the \emph{dense distance graph} (DDG) of a piece $P$ as a directed bipartite graph with vertices $\partial P$ and an edge from every vertex $u \in \top{P}$ to every vertex $v \in \bot{P}$ with weight equal to the length of the shortest $u$-to-$v$ path in $P$.
We denote this graph as $\DDG_P$.\footnote{For general planar graphs,  the DDG of a piece is usually defined as a complete directed graph on $\partial P$.}
$\DDG_P$ can be computed in time $\cO((|\partial P|^2 + |P|) \log |P|)=\cO(|P|\log|P|)$ using the MSSP data structure.
In their seminal paper, Fakcharoenphol and Rao~\cite{FR} designed an efficient implementation of Dijkstra's algorithm on any union of DDGs---this algorithm is nicknamed FR-Dijkstra.
FR-Dijkstra exploits the fact that, due to planarity, the adjacency matrix of each DDG can be decomposed into Monge matrices (defined formally in equation~\eqref{eq:monge} below).
In our case, since each DDG is a bipartite graph, the entire adjacency matrix is itself Monge (this will be shown below).
Let us now give an interface for FR-Dijkstra that is convenient for our purposes.

\begin{theorem}[\cite{FR,DBLP:journals/talg/KaplanMNS17,DBLP:conf/esa/MozesW10}]\label{thm:FR}
Dijkstra's algorithm can be run on the union of a set of $DDG$s with $\cO(M)$ vertices in total (with multiplicities)
and an arbitrary set of $\cO(M)$ extra edges in the time required by $\cO(M \log^2 M)$ accesses to edges of this union.
\end{theorem}

\begin{remark}
In our case, the runtime of the algorithm encapsulated in the above theorem can be improved to $\cO(M\log\log(nw))$.
One of the two $\cO(\log M)$ factors stems from the decomposition of the adjacency matrix into Monge submatrices, which is not necessary in our case.
The second $\cO(\log M)$ comes from the use of binary heaps. In our case, these heaps store integers in $\cO(nw)$ and can be thus implemented with $\cO(\log\log(nw))$ update and query times using an efficient predecessor structure~\cite{DBLP:journals/ipl/Boas77,DBLP:journals/ipl/Willard83}.
\end{remark}

\subparagraph{A warmup.} 
Let us first show how to construct in $\cOtilde(N)$ time an $\cOtilde(N)$-size oracle that answers queries in $\cOtilde(\sqrt{N})$ time using well-known ideas~\cite{FR}.
We will then improve the size of the data structure by efficiently storing the computed DDGs.

Let us consider an $r$-division of $G$, for an $r$ to be specified later.
Further, consider the tree $\AG'$, obtained from the recursive decomposition tree $\AG$ by deleting all descendants of pieces in the $r$-division.
For each piece $P \in \AG'$, we compute and store $\DDG_P$.
In each of the $\cO(\log N)$ levels of $\AG$, for some value $y$, we have $\cO(N/y)$ pieces, each with $\cO(y)$ vertices and $\cO(\sqrt y)$ boundary vertices.
Hence, both the construction time and the space occupied by these DDGs are $\cOtilde(N)$.

We next show how to compute the weight of an optimal alignment of $S[i \dd j]$ and $T[a \dd b]$, i.e.~compute the shortest path $\rho$ from $u=(i,a)$ to $v=(j,b)$, where $i<j$ and $a<b$.
If $u$ and $v$ belong to $P \setminus \partial P$ for a piece $P$ of the $r$-division, then both $S[i \dd j]$ and $T[a \dd b]$ are of length $\cO(\sqrt r)$, and we can hence run the textbook dynamic programming algorithm which requires $\cO(r)$ time.
Henceforth, we consider the complementary case.

Let $P_u$ and $P_v$ be the distinct $r$-division pieces that contain $u$ and $v$, respectively.
Further, let $Q$ be the lowest common ancestor of $P_u$ and $P_v$ in $\AG'$. For $z \in \{u,v\}$,
let $Q_z$ be the child of $Q$ that contains $z$.
The set of vertices $Q_u \cap Q_v$ are denoted by $\sep(Q)$---which stands for separator.
Observe that $\rho$ must contain at least one vertex from $\sep(Q)$.
Consider the set that consists of $P_z$ and the siblings of weak ancestors of $P_z$ in $\AG'$ that are descendants of $Q$, and call it the \emph{cone of $P_z$}.
The cone of $P_z$ covers $Q_z$ and its elements are pairwise internally disjoint.
See~\cref{fig:cones} for an illustration.
Now, observe, that any shortest path $\rho$ between a vertex of $\partial P_z$ and a vertex of $\sep(Q)$ can be partitioned into subpaths $\rho_1, \ldots , \rho_k$ such that each $\rho_i$ lies entirely within some piece $R_i$ in the cone of $z$ and both $\rho_i$'s endpoints are boundary vertices of $R_i$.
Using these two observations, we can compute a shortest $u$-to-$v$ path by running FR-Dijkstra on the cones of $P_u$ and $P_v$, and, possibly, the following extra edges. 
In the case where the source $u$ (resp.~target $v$) is not a boundary vertex, we include $\cO(\sqrt{r})$ additional edges:
for each boundary vertex $x$ of $P_u$ (resp., $P_v$), an edge from $u$ to $x$ (resp., from $x$ to $v$) with length equal to that of the shortest path from $u$ to $x$ (resp.~from $x$ to $v$).
The weights of such edges can be computed in $\cO(r)$ time using dynamic programming.
Thus, a query can be answered in time $\cOtilde(\sqrt{N} + r)$.
By setting $r=\sqrt{N}$ we get the promised complexities.

\begin{figure}[h]
\begin{center}
\includegraphics[page=8,scale=0.5,trim={0 8cm 0 0},clip]{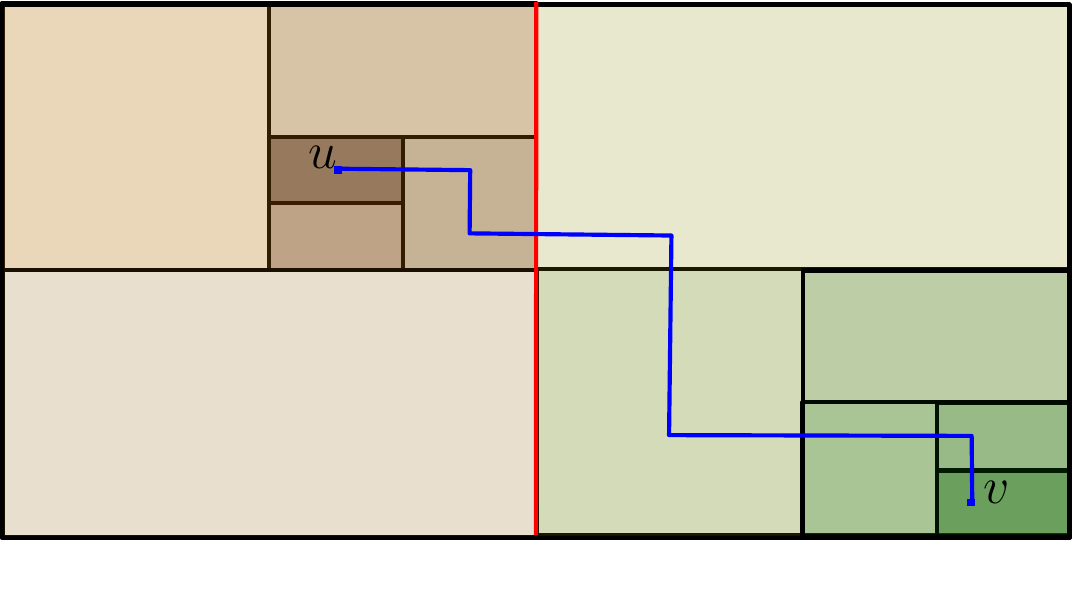}
\caption{The piece $Q$ is shown. $\sep(Q)$ is denoted by red, while a shortest $u$-to-$v$ path is shown in blue. The pieces in the cone of $P_u$ are shaded by brown, while the pieces in the cone of $P_z$ are shaded by pink.}\label{fig:cones}	
\end{center}
\end{figure}

\subparagraph{The tradeoff.} 
We can now describe the entire tradeoff of~\cref{thm:frbased}. 
We assume that $r> w^2$, since otherwise $Nw/\sqrt r = \Omega (N)$ and~\cref{thm:frbased} is  satisfied by the warmup solution. 
For a piece $P$, we will show how to store $\DDG_P$ in $\cO(|\partial P|\cdot w) = \cO(w\sqrt{|P|})$ instead of $\cO(|P|)$ space. 
Our representation will allow retrieving the length of any edge of $\DDG_P$ in $\cOtilde(1)$ time.
Our approach closely follows ideas from~\cite{DBLP:conf/soda/AbboudGMW18}.

For the remainder, we deviate from our ordering convention of $\top{P}$; the first vertex is now the top-right vertex of $P$, and the last is bottom-left. $\bot{P}$ is ordered as before where the first vertex is bottom-left and the last is top-right. 
We denote the $i$-th vertex of $\top{P}$ by $v_i$ and the $j$-th vertex of $\bot{P}$ by $u_j$.
Note that we can infer whether any vertex $u_j$ is reachable from a vertex $v_i$ in $\cO(1)$ time.
For ease of presentation we would like the weights of all edges of $\DDG_P$ to be finite. 
To achieve this, for each edge between two vertices of $\top{P}$, we introduce an artificial edge with weight $w$ in the opposite direction.
It is readily verified that all $v_i$-to-$u_j$ distances that were finite before the introduction of such edges remain unchanged. 
This is because the shortest path between two vertices of this modified graph that lie on the same column (resp.~row) consists solely of vertical (resp.~horizontal) edges.

Let $M$ be the adjacency matrix of $\DDG_P$ with entry $M[i,j]$ storing the distance from $v_i$ to $u_j$, and let $k=|\top{P}|=|\bot{P}|$.
Matrix $M$ satisfies the Monge property, namely:
\begin{equation}\label{eq:monge}
M[i+1,j] - M[i,j] \le M[i+1,j+1] - M[i,j+1] 
\end{equation}
for any $i\in [1,k-1]$ and $j\in [1,k-1]$. This is because the shortest $v_i$-to-$
u_j$ and $v_{i+1}$-to-$u_{{j+1}}$ paths must necessarily cross.

In addition, for any fixed $j\in [1,k]$, for all $i \in [1,k-1]$, we have
\begin{equation}\label{eq:wmonge}
|M[i+1,j]-M[i,j]|\leq w.
\end{equation}
This is because edges $v_iv_{i+1}$ and $v_{i+1}v_i$ both have weight at most $w$.
This implies that $M[i,j] \leq M[i+1,j]+w$, as a shortest $v_i$-to-$u_j$ path cannot be longer than the concatenation of the edge $v_iv_{i+1}$ with a shortest $v_{i+1}$-to-$u_j$ path. Similarly, we have $M[i+1,j] \leq M[i,j]+w$.

Our representation of $M$ is as follows, and fairly standard~\cite{DBLP:conf/soda/AbboudGMW18,DBLP:journals/corr/abs-0707-3619,DBLP:conf/cpm/Charalampopoulos20a}.
We define a $(k-1) \times (k-1)$ matrix $P$, satisfying \[P[i,j]=M[i,j] + M[i+1,j+1] - M[i,j+1] - M[i+1,j].\]
Equations~\eqref{eq:monge} and~\eqref{eq:wmonge} imply that, for any $i\in [1,k-1]$, the sequence of differences
$M[i+1,j]-M[i,j]$
is nondecreasing and contains only values in $[-w,w]$.
Hence, $P$ has $\cO(kw)$ non-zero entries.
Now, observe that 
\begin{equation}\label{eq:permutation}
\sum_{r \geq i, c \geq j} P[r,c]=M[i,j] + M[k,k] - M[i,k] - M[k,j].
\end{equation}

We store the last row and column of $M$. By~\eqref{eq:permutation}, this means that retrieving $M[i,j]$ boils down to computing $\sum_{r \geq i, c \geq j} P[r,c]$. 
We view the non-zero entries of $P$ as points in the plane and build in $\cOtilde(kw)$ time an $\cOtilde(kw)$-size 2D-range tree over them~\cite{Bentley1979}, which can return $\sum_{r \geq i, c \geq j} P[r,c]$ for any $i,j$ in $\cOtilde(1)$ time.
The overall space required by our representation of $\DDG_P$ is thus $\cOtilde(kw)=\cOtilde(|\partial P| \cdot w)$, and any entry of $M$ can be retrieved in $\cOtilde(1)$ time.

In total, over all $\cO(N/r)$ pieces of the $r$-division, the space required is $\cOtilde((N/r)\cdot \sqrt{r}\cdot w)=\cOtilde(Nw/ \sqrt{r})$.
This level dominates the other levels of the decomposition, as the sizes of pieces, as well as their boundaries, decrease geometrically in each root-to-leaf path.
Note that, for the dynamic programming part of the query algorithm, we can simply store the strings, which take $\cO(m+n)$ space.
This concludes the proof of~\cref{thm:frbased}.

\end{document}